\documentclass[conference]{IEEEtran}
\usepackage{graphicx}
\usepackage{textcomp}
\usepackage{xcolor}
\usepackage{xurl}
\usepackage{comment}
\usepackage[british]{babel}

\usepackage{amsmath,amssymb,amsfonts,mathrsfs}
\usepackage{mathtools}
\usepackage[amsmath,thmmarks]{ntheorem}

\usepackage{varioref}
\usepackage{hyperref}
\usepackage[capitalise]{cleveref}

\usepackage{algorithm}
\usepackage{algorithmic}

\usepackage{xspace}
\usepackage{cancel}
\usepackage{enumitem}

\makeatletter
\let\MYcaption\@makecaption
\makeatother
\usepackage[font=footnotesize]{subcaption}
\makeatletter
\let\@makecaption\MYcaption
\makeatother

\newcommand{\By}{\mathbb{B}\kern -0.1em\raisebox{0.5ex}{\scalebox{0.7}{$\mathbb{Y}$}}\xspace}

\renewcommand{\epsilon}{\ensuremath\varepsilon}

\renewcommand{\phi}{\ensuremath{\varphi}}

\theoremstyle{plain}
\numberwithin{equation}{section}

\newtheorem{theorem}{Theorem}[section]

\newtheorem{lemma}[theorem]{Lemma}

\newcommand{\xclaim}{\textsc{Xclaim}\xspace}
\newcommand{\zclaim}{\textsc{Zclaim}\xspace}

\newcommand\zcashvar[1]{\ensuremath{\mathsf{#1}}\xspace}
\newcommand\zclaimstate[1]{\ensuremath{\mathsf{#1}}\xspace}
\newcommand\zcashdef[1]{\ensuremath{\mathbf{#1}}\xspace}

\newcommand\mathvar[1]{\ensuremath{\mathit{#1}}\xspace}
\newcommand\dlt[1]{\ensuremath{\Delta_{#1}}\xspace}

\newcommand{\dvf}{\zcashvar{d}}

\newcommand{\pkd}{\zcashvar{pk_d}}

\newcommand{\dpa}{\ensuremath{(\dvf,\, \pkd)}\xspace}

\newcommand{\rcm}{\zcashvar{rcm}}

\newcommand{\val}{\zcashvar{v}}

\newcommand{\vtot}{\zcashvar{v_{tot}}}
\newcommand{\vmax}{\zcashvar{v_{max}}}

\newcommand{\nlock}{\zcashvar{n_{permit}}}

\newcommand{\sstd}{\ensuremath{\sigma_{std}}\xspace}

\newcommand{\n}{\zcashdef{n}}

\newcommand{\statevr}{\zclaimstate{VaultRegistered}}
\newcommand{\stateni}{\zclaimstate{NotIssuing}}
\newcommand{\statenr}{\zclaimstate{NotRedeeming}}

\newcommand{\statevai}{\zclaimstate{IssueStart}}
\newcommand{\stateam}{\zclaimstate{AwaitingMint}}

\newcommand{\stateaic}{\zclaimstate{Await\-Issue\-Confirm}}

\newcommand{\stateic}{\zclaimstate{IssueChallenged}}
\newcommand{\stateis}{\zclaimstate{IssueSuccess}}

\newcommand{\statevar}{\zclaimstate{RedeemStart}}
\newcommand{\statearc}{\zclaimstate{AwaitRedeem\-Confirm}}

\newcommand{\staterc}{\zclaimstate{RedeemChallenged}}
\newcommand{\staters}{\zclaimstate{RedeemSuccess}}

\newcommand{\statess}{\zclaimstate{StartingState}}
\newcommand{\stateos}{\zclaimstate{OutputState}}

\newcommand{\requestLockop}{\mathvar{requestLock}}
\newcommand{\lockop}{\mathvar{lock}}
\newcommand{\mintop}{\mathvar{mint}}
\newcommand{\burnop}{\mathvar{burn}}
\newcommand{\challengeIssueop}{\mathvar{challengeIssue}}
\newcommand{\challengeRedeemop}{\mathvar{challengeRedeem}}
\newcommand{\releaseop}{\mathvar{release}}
\newcommand{\confirmIssueop}{\mathvar{confirmIssue}}
\newcommand{\confirmRedeemop}{\mathvar{confirmRedeem}}
\newcommand{\submitPOBop}{\mathvar{submitPOB}}
\newcommand{\submitPOCop}{\mathvar{submitPOC}}
\newcommand{\submitPOIop}{\mathvar{submitPOI}}

\newcommand{\lock}{\zclaimstate{lock}}
\newcommand{\mint}{\zclaimstate{mint}}
\newcommand{\burn}{\zclaimstate{burn}}

\newcommand{\release}{\zclaimstate{release}}

\newcommand{\dm}{\dlt{\mathvar{mint}}}
\newcommand{\dci}{\dlt{\mathvar{confirmIssue}}}
\newcommand{\dcr}{\dlt{\mathvar{confirmRedeem}}}

\newcommand{\iw}{\ensuremath{i_w}\xspace}
\newcommand{\issuer}{\textsc{issuer}\xspace}
\newcommand{\redeemer}{\textsc{redeemer}\xspace}
\newcommand{\vault}{\textsc{vault}\xspace}

\newcommand{\oxr}{\ensuremath{\mathcal{O}_{\mathsf{xr}}}\xspace}

\makeatletter
\newcommand{\linebreakand}{%
  \end{@IEEEauthorhalign}
  \hfill\mbox{}\par
  \mbox{}\hfill\begin{@IEEEauthorhalign}
}
\makeatother

\begin{document}

\title{Bridging Sapling: Private Cross-Chain Transfers}
\author{
    \IEEEauthorblockN{Aleixo Sanchez\IEEEauthorrefmark{1}\IEEEauthorrefmark{2}\\
        \href{mailto:aleixo@web3.foundation}{aleixo@web3.foundation}}
        \IEEEauthorblockA{\\[\baselineskip]
        \IEEEauthorrefmark{1}\textit{ETH Z\"urich}\\
        Z\"urich, Switzerland}
    \and
    \IEEEauthorblockN{Alistair Stewart\IEEEauthorrefmark{2}\\
        \href{mailto:stewart.al@gmail.com}{stewart.al@gmail.com}}
    \IEEEauthorblockA{\\[\baselineskip]
        \IEEEauthorrefmark{2}\textit{Web3 Foundation}\\
        Zug, Switzerland}
    \and
    \IEEEauthorblockN{Fatemeh Shirazi\IEEEauthorrefmark{3}\\
        \href{mailto:fatemeh@heliax.dev}{fatemeh@heliax.dev}}
    \IEEEauthorblockA{\\[\baselineskip]
        \IEEEauthorrefmark{3}\textit{Heliax}\\
        Zug, Switzerland}
}

\maketitle

\begin{abstract}

Interoperability is one of the main challenges of blockchain technologies, which are generally designed as self-contained systems. Interoperability schemes for privacy-focused blockchains are particularly hard to design: they must integrate with the unique privacy features of the underlying blockchain so as to prove statements about specific transactions in protocols designed to obfuscate them. This has led to users being forced to weaken their privacy, e.g. by using centralised exchanges, to move assets from one chain to another.
We present \zclaim, a framework for trustless cross-chain asset migration based on the Zcash privacy-protecting protocol. \zclaim integrates with an implementation of the Sapling version of Zcash on a smart-contract capable issuing chain in order to attain private cross-chain transfers. We show that a tokenised representation can be created via a set of collateralised intermediaries without relying on or revealing the total amount to any third party.
\end{abstract}


\section{Introduction}
\label{sec:intro}
Since the creation of Bitcoin in 2009, thousands of cryptocurrencies and blockchains with a wide array of applications have emerged.
This fragmentation of the blockchain space has created a market for interoperability solutions~\cite{schulte2019towards,belchior2021survey}, which we loosely define here as the ability to verify a transaction in a network other than the one in which it was created.
Unfortunately, the vast majority of blockchains are siloed, not having been designed with interoperability in mind.
This results in significant complexity required to design cross-chain protocols; hence, this challenge has until recently mostly been circumvented through the use of centralised exchanges.
These require trust and undermine anonymity, which is particularly undesirable in the case of privacy-oriented cryptocurrencies designed around the notion of payment anonymity.
Despite a mounting interest in recent years in the deployment of privacy features in blockchain protocols, in particular of various zero-knowledge proof primitives~\cite{zhang2019security}, little attention has been given to maintaining the privacy achieved in this manner across chains.
That means users must reveal, or at least risk revealing, their identity if they wish to move assets from one privacy-preserving system to another.
This acts as a major limitation in the usability of these systems, as one can only durably rely on the privacy guarantees of a given system by remaining confined to it.

To address this, in this work we introduce \zclaim \footnote{The \zclaim protocol is the contribution of a master's thesis submitted to ETH Z\"urich in collaboration with Web3 Foundation~\cite{sanchez2020confidential}.
The thesis is accessible under \url{https://github.com/alxs/zclaim} and contains a more exhaustive and technical specification of the protocol.}, a protocol enabling private transactions across chains.
This is to the authors' knowledge the first scheme that succeeds in maintaining privacy in cross-chain transfers. Our protocol follows the structure of \xclaim~\cite{zamyatin2019xclaim}, a framework achieving collateralised trustless asset migration by leveraging smart contract logic, a dynamic set of economically incentivised, trustless intermediaries and cross-chain state verification.
This approach to interoperability falls under what is colloquially referred to as the ``wrapping'' of assets.

\zclaim builds upon said framework by integrating with the Zcash~\cite{hopwood2016zcash} protocol in such a way as to maintain anonymity.
On the chain on which the wrapped assets are issued, the \emph{issuing chain}, we assume an implementation of Sapling according to its specification~\cite{hopwood2016zcash} and introduce new transfer types, with accompanying zk-SNARKs~\cite{banerjee2020demystifying}, to facilitate interoperability.
These transfers integrate with Zcash's private or \emph{shielded} payment scheme and enable the issuing and redeeming of value.
Furthermore, we discuss the custom logic necessary to carry out the issue and redeem protocols, and present and analyse a strategy to hide the transferred amount from intermediaries.

The end result is a protocol enabling trustless cross-chain transfers with similar privacy guarantees to those of Zcash itself, which is considered to provide some of the strongest in blockchain~\cite{zhang2019security,feng2019survey}.

The remainder of the paper is organized as follows.
\cref{sec:background} studies related work and provides a brief summary of concepts in the Zcash protocol.
\cref{sec:protocol} introduces a high-level description of the \zclaim protocol.
\cref{sec:analysis} discusses risks and attacks, and \cref{sec:privacy-inference} presents and analyses an approach to maintain privacy against intermediaries.
\cref{sec:future} considers current limitations and future work and, finally, \cref{sec:conclusion} concludes the paper.

\section{Background}
\label{sec:background}

\subsection{Related Work}
\label{sec:relatedwork}

Atomic swaps~\cite{herlihy2018accs}, the traditional approach to decentralised cross-chain exchanges, allow users on two blockchains to swap ownership of a pre-agreed amount of assets, guaranteeing that the exchange either happens in full or not at all. Atomic swaps present limitations such as 1) requiring the establishing of an external communication channel to find and agree on a swap, 2) an asymmetrical advantage for one out of the two participants (``free option problem''~\cite{Lightnin81:online}), and 3) relatively long confirmation delays. 
There are ongoing efforts to implement atomic swaps between the privacy coin Monero~\cite{van2013cryptonote} and Bitcoin~\cite{BTCXMRatomicswaps,CCSMoneroAtomicSwapsimplementationfunding} as well as Ethereum~\cite{smoothie}.

An alternative approach are \emph{asset migration} protocols~\cite{zamyatin2019sok} such as the one presented in this paper, in which a representation of an asset (its \emph{tokenised representation} or ``wrapped'' version) is created on a different chain, while those on the original chain are locked until the process is reversed.

At least two different projects---Wrapped~\cite{Wrapped}, a ``provider of wrapped layer-1 assets'', and the Ren Project, a generic cross-chain transfer protocol~\cite{HomerenprojectrenWikiGitHub}---offer tokenised representations of Zcash on Ethereum~\cite{Zcash76:online}.
However, the former relies on a centralised authority and neither of them supports Zcash's shielded payment scheme.
The Zcash Foundation announced~\cite{pegzone_announ} their intent to work on a ``Zcash pegzone''~\cite{githubPegzone} for the Cosmos~\cite{cosmosWhitepaper} ecosystem in 2020, with the goal of enabling ``shielded transfers from the pegzone to Zcash and vice versa''; however, details on the project are scarce.

\subsection{Zcash and Sapling}

We now introduce basic concepts from Zcash used in this work.
For a formal definition of these concepts and the cryptographic schemes employed in them, we refer the reader to the Sapling version of the Zcash protocol specification~\cite{hopwood2016zcash}.
Some terms have been simplified for ease of understanding.

Zcash is an implementation of the Zerocash~\cite{sasson2014zerocash,sasson2014zerocash_ext} payment scheme.
It builds on Bitcoin's~\cite{nakamoto2008bitcoin} transparent payment scheme, adding a private payment scheme leveraging zk-SNARKs~\cite{banerjee2020demystifying} to enable private payments.
This work only concerns itself with the latter.

Transactions in Sapling can contain transparent inputs, outputs, and scripts, all of which work as in Bitcoin~\cite{nakamoto2008bitcoin}, and shielded \emph{JoinSplit}\footnote{We ignore JoinSplit transfers, which Sapling only supports for backwards compatibility.}, \emph{Spend} and \emph{Output transfers}.
Spend and Output transfers are analogous to transparent inputs and outputs, respectively. Each Spend transfer spends a \emph{note}, and each Output transfer creates one.
A note represents that a value \val is spendable by the recipient who holds the private key to the destination \emph{shielded payment address}. A note's sender, recipient and value are never revealed.

All Spend and Output transfers in a transaction, along with any transparent inputs and outputs, are checked to balance by verifying that the sum of all value commitments and of all transparent values is equal to zero.

To each note there is a cryptographically associated \emph{note commitment}, which is added to the \emph{note commitment tree} when the note is created. Only notes whose note commitment is in the note commitment tree can be spent. When the note is spent, a \emph{nullifier} uniquely associated with that note must be revealed and is then added to the \emph{nullifier set}. It is infeasible to compute the nullifier without the spending key corresponding to the recipient's shielded payment address, and only notes whose nullifier is not in the nullifier set can be spent.

The main premise of Zcash's shielded payment scheme is that when a note is spent, the spender only proves that its note commitment is in the note commitment tree. Revealing its nullifier also does not reveal which note commitment it is associated with, which means that a spent note cannot be linked to the transaction in which it was created.

The values required to spend a note are generally encrypted to its recipient when it is created, though the protocol does not enforce it.
This has been a challenge in designing \zclaim, since it is not enough to verify that a note exists in the note commitment tree in order to verify a transaction.
Instead, we must also ensure that the counterparty receives these values.

\section{\zclaim Protocol}
\label{sec:protocol}

We introduce \zclaim, an adaptation of \xclaim to Zcash so as to facilitate private transfers to any blockchain that supports custom logic and efficient verification of the required cryptographic functions~\cite[Section 5.2]{sanchez2020confidential}.

We define the following terms:
\begin{itemize}
    \item \textbf{ZEC} denotes the Zcash cryptocurrency.
    \item \textbf{wZEC} denotes wrapped ZEC, the tokenised representation of ZEC on the issuing chain.
    \item \textbf{$I$} is the issuing chain, on which wZEC will be created.
    \item \textbf{$i$} is an implementation-defined existing currency on $I$.
\end{itemize}

\subsection{Actors}

The following actors participate in the protocol:
\begin{itemize}
    \item \textbf{Issuers} lock ZEC on Zcash to request an equivalent amount of wZEC on $I$.
    \item \textbf{Redeemers} destroy wZEC on $I$ to request the corresponding amount of ZEC on Zcash.
    \item \textbf{Vaults} are the non-trusted intermediaries that act as custodians, safekeeping locked ZEC.
    Anyone can take on the role of a vault by locking some collateral in $i$ and registering as such.
    Vaults are liable for fulfilling redeem requests of wZEC for ZEC.
    They are incentivised by fees they derive from transactions and in case of misbehaviour, they face partial or total liquidation of their collateral.
\end{itemize}

\subsection{Components}

The functionality that \zclaim requires on the issuing chain can be split into the following components.

\subsubsection{Vault registry}

The vault registry keeps a public list of all registered vaults and their status.
Each vault has a Zcash shielded payment address $(d, pk_d)$ and an amount of collateral associated with it.
The collateral it keeps in $i$ on the issuing chain and thus the total amount of ZEC it is able to accept are public; however, the amount of ZEC \emph{obligations} to its name, i.e. ZEC it participated in issuing and is accountable for releasing, is not.
It periodically proves that this amount is properly backed by its collateral by submitting \emph{proofs of balance}, which consist of a zk-SNARK accounting for all previous requests it has served.

\subsubsection{Relay system}
\label{sec:relay_system}
The relay system~\cite{buterin2016interop} keeps track of the state of the Zcash chain.
Specifically, it verifies and stores block headers, provides a mechanism to signal that consensus has been reached on a given block similarly to an SPV or light client~\cite{SPVBitcoinWiki,Back2014sidechains}, and allows the verification of Zcash notes.
Reaching consensus, in this context, requires fulfilling implementation-specific criteria w.r.t. Zcash's probabilistic finality; a possible approach is discussed in \cref{sec:relay_poisoning}.
In order to prove the existence of a note, users must provide an \emph{inclusion proof}, which consists of a Merkle path from its note commitment to the root of the note commitment tree in the block header of a previous block.
On the other hand, the relay system does not need to offer functionality to verify whether a nullifier already exists in the nullifier set.

Block headers are submitted to the relay system by \emph{relayers}, which may or may not be economically incentivised.
Other protocol actors are likely to take on the additional role of relayers if the cost associated with doing so is low and it is in their interest to guarantee the stability of the protocol.
Vaults, for example, are in such a position: they need to run full nodes of both chains, and the financial damage they may incur from an attack on the relay system as discussed in \cref{sec:relay_poisoning} is far greater than the cost of submitting headers.

\subsubsection{Exchange rate oracle}
The exchange rate oracle \oxr provides an exchange rate that reflects the prevailing market value of 1 ZEC in $i$.
The design of the exchange rate oracle falls outside the scope of this work, but a brief discussion on exchange rate sourcing and security concerns can be found in \cref{sec:er_poisoning}.

\subsubsection{Protocol logic}
The logic and state pertaining to the issue and redeem protocols may but need not be contained in a set of smart contracts.
Alternative approaches to user-defined on-chain logic may allow these to be stored e.g. in the runtime~\cite{burdges2020overview} or as precompiled modules~\cite{dfinity2022internet}.

\subsection{Transfers and transactions}

Along with Sapling's Spend and Output transfers, the issuing chain supports Mint and Burn transfers.
A Mint transfer is a transfer creating value of the issued currency, i.e. increasing the circulating supply.
A Burn transfer takes the burnt amount out of circulation.
Transactions in the issuing chain may consist of any of the same components as a Sapling transaction, and either a Mint or a Burn transfer.
A transaction containing a Mint transfer is a \mint transaction, and one containing a Burn transfer a \burn transaction.
Mint transfers contain a zk-SNARK proving that a note has been sent to a vault on Zcash and that the minted value corresponds to the locked value. 
In Burn transfers, redeemers create a note they wish to receive on Zcash and encrypt it to the vault, while publishing its note commitment.
The zk-SNARK in Burn transfers allows redeemers to show that the value of said note matches the burnt value.
The transaction is only confirmed once the vault provides an inclusion proof for this note.
We refer the reader to the original specification of this protocol~\cite[Section 5]{sanchez2020confidential} for concrete definitions of the aforementioned transfers and of zk-SNARKs mentioned in this work.

\subsection{Issuing and Redeeming}
\label{sec:protocols_highlevel}

We define two sub-protocols, Issue and Redeem, adapted from \xclaim to address the challenges arising from integrating Zcash's anonymity features.

The Issue sub-protocol allows issuers to lock ZEC with vaults and mint wZEC on the issuing chain.
In order to conceal the transferred amount from the vaults, the issuer splits the total into separate amounts beforehand and sends them to $n$ vaults such that no vault is able to deduce the total.
We present and analyse a possible \emph{splitting strategy} to this end in \cref{sec:splitting_strategy}.
For each of these amounts, the issuer carries out an Issue \emph{procedure}, i.e. an instance of the Issue sub-protocol.
In short, the issuer proves that they have created a note encrypted to a vault's address and mints an equivalent amount of wZEC in a \mint transaction.
The existence of this note does not guarantee that it can be spent by the vault, hence there is a challenge period during which vaults can void the transaction by proving that the necessary values have not been properly encrypted to them.

The Redeem sub-protocol allows a redeemer to receive ZEC from vaults in exchange for burning wZEC on the issuing chain. 
Redeemers shall employ the same splitting strategy as issuers.
For each individual amount, a redeemer submits a \burn transaction on the issuing chain which contains a commitment to a note of equivalent value spendable by an address under their control on Zcash.
A vault then creates this note.
Since the vault must also learn the values from which the note is derived in order to create it, there is a challenge period during which it can challenge the request by showing that these were not properly encrypted.
If the values are correct, the vault must create the note specified by the redeemer.
Lastly, the vault submits a proof showing that the note exists and the \burn transaction is finalised.

We present an overview of both sub-protocols.
Diagrams depicting these steps are shown in \cref{fig:zclaim_diagrams}.

\begin{figure}
\centering
\vspace{.2\bigskipamount}
\hspace{1.3\bigskipamount}
\begin{subfigure}[b]{.4\textwidth}
  \includegraphics[width=\linewidth]{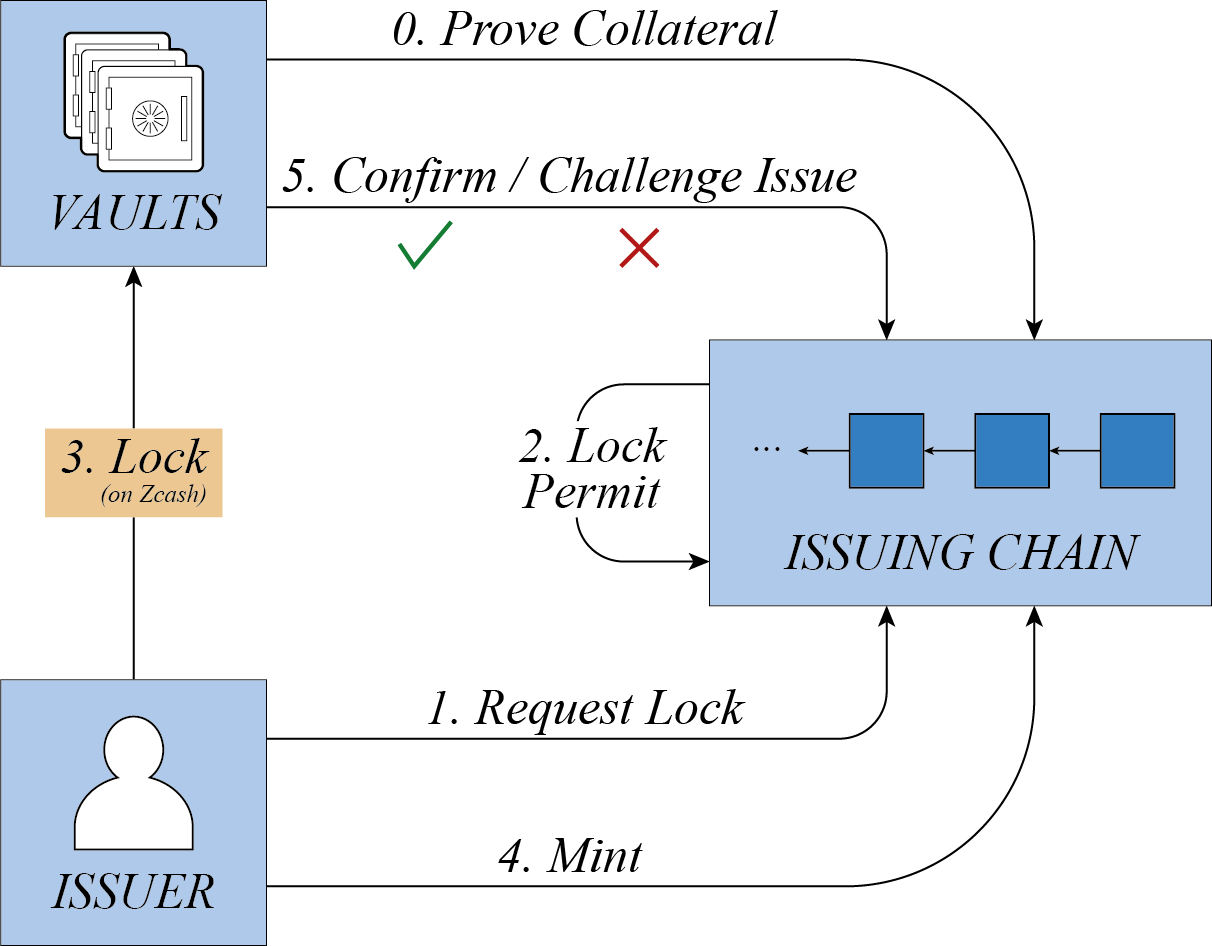}
  \subcaption{Issue}
\end{subfigure}%
\par\vspace{1.2\bigskipamount}
\hspace{1.3\bigskipamount}
\begin{subfigure}[b]{.4\textwidth}
  \includegraphics[width=\linewidth]{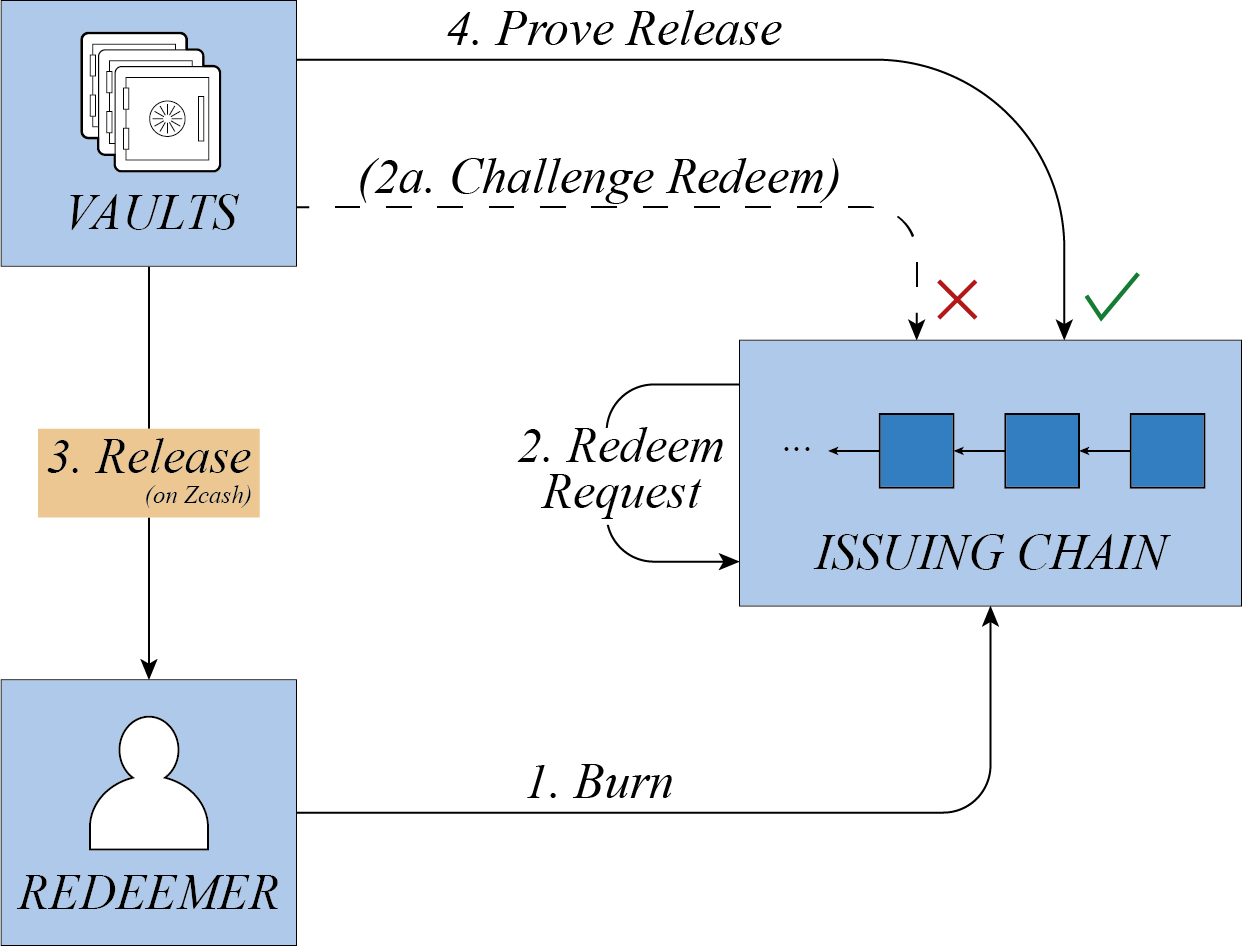}
  \caption{Redeem}
\end{subfigure}
\caption{Simplified diagrams of \zclaim protocols.}
\label{fig:zclaim_diagrams}
\end{figure}

\textbf{Issuing:} Alice (issuer) locks funds with vault V on Zcash to create units of wZEC on $I$.
\begin{enumerate}[start=0]
    \item \emph{Setup.} V registers with the vault registry on $I$ and locks  $i_{col}$ units of collateral, where it must hold that
    \begin{equation}\label{eq:collateral_issue}
        i_{col} \geq \vmax \cdot (1 - f) \cdot \sstd \cdot xr_{cap}
    \end{equation}
    where \vmax is the maximum amount of funds that can be locked or burned per request, $f$ is the \zclaim transaction fee, and $\sstd$ is the standard collateralisation rate.
    These values are implementation-defined constants.
    $xr_{cap}$ is the exchange rate at the time the vault provides a proof of capacity as explained next.
    Further, V must provide a shielded payment address on Zcash to which issuers will send their funds and submit a \emph{proof of capacity} to the vault registry, which is a zk-SNARK proving that \cref{eq:collateral_issue} holds.
    V must periodically resubmit proofs of capacity in order to remain available for requests.

    \item \emph{Commit.} Alice makes a request to the issuing chain to lock her funds with vault V with shielded payment address $\dpa_\text{V}$.
    As part of this request, Alice locks a small amount of $i$, \iw, her \emph{warranty collateral}.
    This is a fixed amount large enough to compensate V for the opportunity cost in case Alice does not follow through with her request.
    Furthermore, it serves to prevent griefing attacks in which the system is spammed with lock requests, becoming unavailable to legitimate users.
    
    \item \emph{Lock permit.} Subsequently, a \emph{lock permit} is created on $I$, granting permission to Alice to lock her funds with V.
    This permit contains a cryptographic nonce $n_{lock}$ that Alice must include in the transaction locking the funds.

    If she fails to execute step 4 within some \dm, her warranty collateral \iw is transferred to V.
    The constant \dm will largely depend on how fast the relay system considers blocks to have reached consensus.
    As a point of reference, popular cryptocurrency exchanges require a depth of 24 blocks or 30 minutes~\cite{DepositProcessingTimesKraken,DepositProcessingTimesGemini}.

    \item \emph{Lock.} Alice creates a shielded transaction on Zcash, sending $ZEC_{lock}$ to V.
    Alice uses $n_{lock}$ to derive the randomness \rcm, the so-called commitment trapdoor, used to generate the commitment for the output note \n addressed to V in this transaction.
    
    \item \emph{Create.} Alice makes a request to issue $wZEC_{create}$ to a shielded address on $I$.
    To this end, Alice provides an inclusion proof for a note with note commitment $cm_\n$, and further proves in zero knowledge that:
    \begin{itemize}
        \item she knows a note \n with note commitment $cm_\n$, recipient $\dpa_\text{V}$ and value $ZEC_{lock}$
        \item $wZEC_{create} = ZEC_{lock} (1 - f)$, where $f$ is the fixed fee rate the vault earns, and $ZEC_{lock} \leq \vmax$
        \item the trapdoor \rcm was derived from $n_{lock}$
    \end{itemize}
    The transaction $T_{\mintop}$ issuing $wZEC_{create}$ remains pending until V confirms it.
    If V fails to do so within some \dci, the same \iw is deducted from V's collateral and transferred to Alice, and $T_{\mintop}$ is confirmed.
    This delay may be quite small, as it only needs to allow for V to see and respond to $T_{\mintop}$ on $I$.
    Furthermore, Alice publishes a note ciphertext $C^{\text{V}}$ of the note \n symmetrically encrypted to V.

    \item \emph{Confirm/Challenge.} V decrypts $C^{\text{V}}$ and verifies whether the resulting note has note commitment $cm_\n$.
    If that is the case, V confirms and the issuing process is complete.

    On the other hand, if V finds that it cannot properly decrypt $C^{\text{V}}$, it may challenge the transaction by revealing the shared secret used in the encryption while proving its correctness in a zk-SNARK.
    It can then be verified on chain that the encryption was erroneous, in which case $T_{\mintop}$ is discarded and Alice loses $ZEC_{lock}$ and \iw.
\end{enumerate}

\textbf{Redeeming:} Dave (redeemer) burns wZEC on $I$ and obtains ZEC from vault V.
\begin{enumerate}[start=0]
    \item \emph{Setup.} V is available to redeem, i.e.\ has not provided a \emph{proof of insolvency} to the vault registry since it last participated in an Issue procedure.
    Vaults are available for redeem requests by default.
    In order to be exempted from them, they need to periodically provide proofs of insolvency, showing that their ZEC obligations are smaller than \vmax.
    This simplified requirement ensures users will always be able to redeem the existing supply.

    \item \emph{Burn.} Dave makes a request to burn funds $wZEC_{burn}$ on $I$ by locking \iw as warranty collateral and submitting a transaction $T_{\burnop}$.

    \item \emph{Redeem request.} In this transaction, Dave specifies that he would like to redeem ZEC from V in a note \n with note commitment $cm_\n$, and proves in zero knowledge that:
    \begin{itemize}
        \item he knows a note \n with note commitment $cm_\n$ and value $ZEC_{release}$
        \item $ZEC_{release} = wZEC_{burn} (1 - f)$, where $f$ is the fixed fee rate that the vault earns, and $wZEC_{burn} \leq \vmax$
    \end{itemize}
    
    In redeem requests, vaults earn fees implicitly through their ZEC obligations decreasing by a larger amount than the ZEC amount they release.
    
    In order to transmit the note values to V, Dave publishes the note ciphertext $C^{\text{V}}$ of the note \n encrypted to V.
    $T_{\burnop}$ is not confirmed until V confirms the release of $ZEC_{release}$.
    If V fails to do so or to challenge it within some \dcr, \iw is deducted from V's collateral and transferred to Dave, and $T_{\burnop}$ is discarded.
    \dcr, like \dm, will largely depend on how fast the relay system accepts blocks.
    
    \item \emph{Release.} V releases $ZEC_{release}$ to Dave by creating a note \n with note commitment $cm_\n$.

    \item \emph{Confirm/Challenge.} V waits until the relay system signals that consensus has been reached on the block in which \n was created and then submits an inclusion proof for this note.
    $T_{\burnop}$ is then confirmed.
    
    However, if upon decryption of $C^{\text{V}}$ V finds that the resulting plaintext does not correspond to a note with note commitment $cm_\n$, it may challenge the transaction as in step 5 of the Issue protocol, in which case $T_{\burnop}$ is voided and Dave's warranty collateral is transferred to V.
    In this case, V does not execute step 3.
\end{enumerate}

\subsection{Operations}
\label{sec:ops}

We now define abstract operations for the two protocols introduced in the previous section.

An operation always results in a transaction.
On $I$, this may either be a monetary transfer or a change in state.
Each operation leads to a specific state of the protocol, and to each state there is an implicitly associated set of legal operations for each party.

We denote by $T_{op}^{A}$ the transaction on blockchain $A$ resulting from the successful execution of operation $op$.
$A$ may be either Zcash, denoted by $Z$, or the issuing chain, denoted by $I$.
Certain transactions on $I$ require confirmation from the other party involved in the Issue or Redeem procedure and may be voided if this confirmation is not provided within a certain delay. 
We denote a transaction that is pending confirmation as $(T_{op}^I)$, and one that is voided as $\cancel{(T_{op}^I)}$.

An operation that results in transaction $T_{op}$, can be performed in state \statess and leads to state \stateos is denoted by $op \rightarrow T_{op}^{A} \, [\statess \rightarrow \stateos]$.
For simplicity, we omit non-monetary and warranty collateral transactions as well as slashing, all of which take place only on $I$.
All operations can only be performed by a specific actor.

Operations performed by issuers:
\begin{itemize}
    \item $\requestLockop \; [\statevai \rightarrow \stateam]$ requests permission to lock funds with a vault.
    The issuer locks \iw as collateral.
    
    \item $\lockop \rightarrow T^Z_{\lockop} \; [\stateam \rightarrow \stateam]$ locks ZEC with a vault.
    
    \item $\mintop \rightarrow (T^I_{\mintop}) \; [\stateam \rightarrow \stateaic]$ allows the issuer to mint a hidden amount of wZEC on $I$ upon confirmation.

    If the vault fails to perform \confirmIssueop within \dci, $T^I_{\mintop}$ is automatically confirmed and \iw is deducted from the vault's collateral.
\end{itemize}

Operations performed by redeemers:
\begin{itemize}
    \item \mbox{$\burnop \rightarrow (T^I_{\burnop}) \; [\statevar \rightarrow \statearc]$} requests a specific vault to release ZEC and burns an equivalent amount of wZEC upon confirmation.
    The redeemer locks \iw as collateral.
    
    If the vault fails to perform \confirmRedeemop within \dcr, $T^I_{\burnop}$ is voided and \iw is deducted from the vault's collateral.
\end{itemize}

Operations performed by vaults:
\begin{itemize}
    \item $\challengeIssueop \; [\stateaic \rightarrow \stateic]$ allows a vault to prove that $T^Z_{\lockop}$ has not been correctly encrypted to it in $(T^I_{\mintop})$.
    
    \item $\challengeRedeemop \; [\statearc \rightarrow \staterc]$ allows a vault to prove that $T^Z_{\releaseop}$ has not been correctly encrypted to it in $(T^I_{\burnop})$.
    
    \item $\releaseop \rightarrow T^Z_{\releaseop} \; [\statearc \rightarrow \statearc]$ releases funds to a redeemer.
    
    \item $\confirmIssueop \; [\stateaic \rightarrow \stateis]$ allows a vault to confirm that it has received $T^Z_{\lockop}$, confirming the pending $(T^I_{\mintop})$ transaction.
    
    \item $\confirmRedeemop \; [\statearc \rightarrow \staters]$ proves that the vault has released funds in $T^Z_{\releaseop}$, confirming the pending $(T^I_{\burnop})$ transaction.
    
    \item
    \begin{sloppypar}
    $\submitPOBop \; [\statevr | \statevai | \allowbreak\stateni \rightarrow \stateni]$ submits a proof of balance.
    Vaults may perform this operation in order to prevent liquidation instead of \submitPOCop if they don't wish to become available for new lock requests.
    \end{sloppypar}
    
    \item
    \begin{sloppypar}
    $\submitPOCop \; [\statevr\allowbreak|\statevai\allowbreak|\stateni \rightarrow \statevai]$ submits a proof of capacity.
    \end{sloppypar}
    
    \item
    \begin{sloppypar}
    $\submitPOIop \; [\statenr\allowbreak|\statevar \rightarrow \statenr]$ submits a proof of insolvency.
    \end{sloppypar}
\end{itemize}

Based on these operations, the Issue and Redeem protocols are summarised in pseudocode in \cref{alg:issue,alg:redeem}, respectively.

\begin{minipage}[t]{.45\textwidth}
\begin{algorithm}[H]
    \caption{Issue}\label{alg:issue}
    \footnotesize
    \begin{algorithmic}[1]
        \REQUIRE \vault has enough collateral
        \STATE \vault executes \submitPOCop
        \IF{\issuer executes \requestLockop and receives a lock permit}
            \STATE \issuer executes $\lockop ^{\rightarrow T^Z_{\lockop}}$ 
            \IF{\issuer executes $\mintop ^{\rightarrow (T^I_{\mintop})}$ within \dm}
                \IF{\vault executes \challengeIssueop within \dci}
                    \STATE $\rightarrow \cancel{T^I_{\mintop}}$
                \ELSE
                    \STATE \vault may execute \confirmIssueop within \dci
                    \STATE $\rightarrow T^I_{\mintop}$
                \ENDIF
            \ENDIF
        \ENDIF
    \end{algorithmic}
\end{algorithm}
\end{minipage}

\begin{minipage}[t]{.45\textwidth}
\centering
\begin{algorithm}[H]
    \caption{Redeem}\label{alg:redeem}
    \footnotesize
    \begin{algorithmic}[1]
        \REQUIRE \vault has not called \submitPOIop since last issuing
        \STATE \redeemer executes $\burnop ^{\rightarrow (T^I_{\burnop})}$
        \IF{\vault executes \challengeRedeemop within \dcr}
            \STATE $\rightarrow \cancel{T^I_{\burnop}}$
        \ELSE
            \STATE \vault executes $\releaseop ^{\rightarrow T^Z_{\releaseop}}$
            \IF{\vault executes \confirmRedeemop within \dcr}
                \STATE $\rightarrow T^I_{\burnop}$
            \ELSE
                \STATE $\rightarrow \cancel{T^I_{\burnop}}$
            \ENDIF
        \ENDIF
    \end{algorithmic}
\end{algorithm}
\end{minipage}

\section{Attack vectors and points of failure}
\label{sec:analysis}

We discuss here a range of attacks and points of failure in \zclaim and offer mitigation strategies.
Where not specified otherwise, the discussion on vulnerabilities presented in the security analysis of \xclaim~\cite[Section~VII]{zamyatin2019xclaim} also holds here.

\subsection{Inference attacks}
\label{sec:inference_attacks}

Without a robust splitting strategy, vaults may guess the users' identity through the amounts in \lock and \release transactions in which they are involved.

In the simplest scenario, if a user locks an amount with a vault which matches a recent transparent-to-shielded transaction, the vault may deduce the origin of the funds~\cite{quesnelle2017linkability} and infer the user's identity from activity associated with the transparent address.

The same attack is possible if a vault is able to deduce or estimate the total from one or a small subset of the amounts sent to it.
The splitting strategy defined in \cref{sec:splitting_strategy} aims to prevent this sort of attack.
The privacy provided by this strategy is analysed in the remainder of said section.

\subsection{Chain relay poisoning}
\label{sec:relay_poisoning}

Chain relay poisoning involves an adversary triggering a chain reorganisation such that a previously accepted transaction is invalidated; in other words, bypassing the relay system's consensus validation mechanism.

Since Zcash has probabilistic finality, this is in theory always possible.
However, a common practice consists in deeming only blocks at depth $h \geq k$ to have reached consensus, where $k$ is a security parameter denoting the block depth at which the likelihood of an adversary with computational power bounded by $\alpha \leq 33\%$~\cite{gervais2016security} triggering a chain reorganisation becomes negligible.
This threshold is usually a low number due to the probability of a successful attack decreasing exponentially with increasing depth.
Nevertheless, it must be noted that the dangerous assumption here is not the threshold $k$ but $\alpha$.
The costs required to reach this threshold vary wildly for different cryptocurrencies and a more in-depth exploration of the feasibility of such an attack on Zcash is recommended in order to find sensible values for these parameters.

Furthermore, a poisoning attack may be successful well below this threshold $\alpha$ if the relay system is deprived of recent block header data~\cite[Section VII-A]{zamyatin2019xclaim}.
This is a common vulnerability to interoperability schemes and may involve intricate attacks in which relayers are isolated from the rest of their peers in the network and misled to accept the attacker's chain as the longest, in what is commonly known as an eclipse attack~\cite{heilman2015eclipse}.
Note that as long as there is one honest relayer connected to the rest of the network, the cost of attacking the relay system is the same as that of running a 51\% attack, since block headers are verified based on proof of work and not on the number of relayers submitting them.

Such attacks along with mitigation strategies have been discussed previously in the literature~\cite{heilman2015eclipse,wust2016ethereum,xu2020eclipsed,alangot2020decentralized}, and we refer the reader in particular to analyses on Bitcoin and Bitcoin-based blockchains (such as Zcash is) for mitigation strategies~\cite{heilman2015eclipse,alangot2020decentralized}.

\subsection{Exchange rate poisoning}
\label{sec:er_poisoning}

If the exchange rate oracle is manipulated to provide an erroneous price feed, the protocol may fail to maintain proper collateralisation of the issued assets.
An artificially high exchange rate would allow vaults to issue ZEC or unlock collateral such that they become undercollateralised after the attack.
On the other hand, an artificially low exchange rate may trigger mass liquidation and allow users to buy vaults' collateral at an unfair price.

It is thus important to guarantee the reliability of the exchange rate oracle.
Blockchain oracles aim to solve this exact problem~\cite{peterson2015augur,ellis2017chainlink}, aggregating exchange rates from different sources, leveraging economic incentives to reinforce their veracity and providing dispute mechanisms in case of discrepancies.
These systems are, without question, safer than relying on a single source to provide an exchange rate, though they may still fail under certain circumstances~\cite{lo2020reliabilityoracles}.

\subsection{Replay attacks on inclusion proofs}

\zclaim prevents replay attacks on \lock and \release transaction inclusion proofs, in which a user reuses a past \lock transaction to mint wZEC or a vault reuses a \release transaction to decrease their ZEC obligations, as follows.

The nonce $\nlock$ generated in lock permits must be used to generate the note commitment trapdoor in \lock transactions, which is enforced in the zero knowledge proof in Mint transfers.
This ensures that every \lock transaction is uniquely associated with the corresponding Issue procedure.
As for \release transactions, protection from replay attacks is implicit since the note commitment is generated in advance by the redeemer.
A vault can only replay a release inclusion proof if the redeemer purposefully chooses the same note values, most notably the same note commitment trapdoor as in a previous Burn transfer.

\subsection{Counterfeiting}
\label{sec:counterfeiting}

We define counterfeiting as the issuing of wZEC which is not backed by an equivalent amount of collateral.
As previously outlined, vaults must periodically provide proofs of balance showing that their ZEC obligations are collateralised above a certain ratio.
If a vault fails to do so and the exchange rate changes by a certain margin since their last statement, partial liquidation of their collateral is triggered.
This ensures the collateralisation ratio never falls below a certain minimum.
Furthermore, proofs of capacity ensure that such a statement will hold after a mint. See~\cite[Section 5.12]{sanchez2020confidential} for more details on balance statements.

Note that wZEC does not need to be fully backed by ZEC, which is impossible to guarantee while maintaining the vaults' transaction history private.
A vault may very well reuse ZEC it has received to issue more wZEC; but it will remain unable to unlock its collateral until it has released ZEC to redeemers or acquired and burnt wZEC itself.

\subsection{Sudden devaluation}
\label{sec:black_swan}

Given the historically high volatility in cryptocurrency markets~\cite{fry2016negative,cheah2015speculative}, the risk of sudden, extreme devaluation of either one of Zcash or the issuing currency w.r.t. the other is non-negligible.
In case of a drop in the valuation of Zcash, \zclaim would continue to operate normally.
In the opposite case, the consequences would be similar to those discussed in \cref{sec:er_poisoning} and \zclaim would eventually no longer function.
The strategy proposed in \xclaim \cite[p. 3]{zamyatin2019xclaim} can be employed to mitigate the risk.

\section{Splitting strategy and analysis of inference attacks}
\label{sec:privacy-inference}

In this section, we demonstrate that vaults cannot infer the total value prior to splitting. We consider a simple threat model that describes what a single vault can infer from each quantity received. We assume the user has an amount $\vtot$ to transfer which is split among $k$ vaults, where $k$ is a low power of two. If any vault receives an amount $v$, the vault knows that $\vtot \geq v$. If the user transfers an amount $\vtot$, then at least one vault will receive $\vtot/k$ and so they can infer that the total amount must be at least $\vtot/k$. This means that at the very high values of $\vtot$, anonymity suffers.

We remark here that transaction fees on Zcash are generally very low\footnote{The default is 0.00001 ZEC~\cite{ZcashFea8:online}, worth around 0.00176 USD at the time of writing.}, hence requiring a total of up to e.g.\ $k = 16$ transactions per transfer is a realistic solution.
Nevertheless, this also needs to be the case on the issuing chain.

\subsection{A scale-independent prior}

The amount received by each vault is used to assign a probability to the total amount using Bayesian inference. This relies on a prior distribution on the total amount a user might wish to transfer. This amount may span many orders of magnitude, but the scale must be kept private. We assume that the total amount is an integer in the range $[1,2^h-1]$. The prior distribution is then sampled as $T=2^N+A$, where $N$ is chosen uniformly from the integers in $[0,h]$, and $A$ is chosen uniformly from the integers in $[0,2^N-1]$. 

\subsection{Proposed splitting protocol} \label{sec:splitting_strategy}

In our splitting protocol, each user splits their total amount $\vtot$ into $k$ pieces. We restrict the size of each piece to $0$ or powers of two. To attribute piece sizes, an obvious protocol is to use the powers of two that correspond to the highest $1$ bits in the binary representation of the total amount $\vtot$. However, this would give away information about the total, since each piece would eliminate half the possible values. We get around this only using pieces of sizes from 0 to $2^m$, where $m=h+1-\log_2 k$. If \vtot is large, many pieces must be of size $2^m$, which unfortunately means there are less pieces to use on bits.

If the total amount $t$ to split is in  $[1, 2^{m}-1]$ then the splitting procedure is as follows:

\begin{itemize}
    \item Let $e=\max \{1, 2^{\lfloor \log_2 \vtot \rfloor +1 - k/2} \}$.
    \item Choose an integer $i$ uniformly at random from $[0,e \lfloor \vtot/e\rfloor]$.
    \item When $e > 1$, do not transfer $\vtot-e\lfloor \vtot/e \rfloor$ tokens.
    \item Split the remaining $e\lfloor \vtot/e\rfloor$ tokens into $e i$ and $e (\lfloor \vtot/e \rfloor-i)$, then split each of those into powers of $2$.
    \item If this gives less than $k$ pieces, then pad with $0$.
\end{itemize}

If the amount \vtot is in $[2^{m+1},2^h-1]$, then the splitting procedure is as follows:

\begin{itemize}
    \item Let $d=\lfloor \vtot/2^m \rfloor-1$, $c=\lfloor (k-d)/2 \rfloor$ and $e=2^{m-c}$.
    \item Split the total amount into $d$ pieces of size $2^m$.
    \item Do not transfer $\vtot-e \lfloor \vtot/e \rfloor$ tokens.
    \item Choose an integer $i$ uniformly at random from $[0,  \lfloor \vtot/e \rfloor - d 2^m/e]$ and split the remaining $e \lfloor \vtot/e \rfloor - d 2^m$ tokens firstly into $ie$ and $e \lfloor \vtot/e \rfloor - d 2^m-ie$. 
    \item Split these into powers of two and pad with $0$s to form $k$ pieces.
\end{itemize}

Note that this never gives the user more than $k$ pieces.

\subsection{Forward probabilities}

Assuming the user has a total of $\vtot$ tokens, we now consider the probability that a vault chosen at random from the $k$ vaults has a piece size of $2^j$, for some $0 \leq j \leq m$. We write $X_0,X_1,\dots,X_{m+1}$ to denote random variables indicating the number of pieces of size $0,1,\dots, 2^m$ respectively. We write $T$ for the random variable of the total amount. The probability that a randomly-chosen vault has a piece size of $2^{j-1}$ for some $1 \leq j \leq m+1$, given that the user had $\vtot$ in total, is then $E[X_j|T=\vtot]/k$.

\begin{lemma} \label{lem:binary-random-split}
If $i$ is selected uniformly at random from the integers in  $[0,2^c+a]$ for some integer $0 \leq a < 2^c$, and $Y_j$ denotes the $j$th bit of i, then:
\begin{enumerate}[label=(\roman*),leftmargin=1cm]
    \item For $0 \leq j \leq c$, $1/4 \leq \Pr[Y_j=1] \leq 3/4$.
    \item $\Pr[Y_{c+1}=1] \leq 1/2$.
    \item The expected number of 1s in the binary expansion of $i$, $E[\sum_j Y_j]$, is between $c/4$ and $(3c+2)/4$.
\end{enumerate}
\end{lemma}

The lemma follows from the fact that selecting $i$ from $[0,2^c-1]$ means that each bit would be $1$ with probability $1/2$, and selecting $i$ from $[0,2^c+a]$ means that $i$ is in $[0,2^c-1]$ with probability over $1/2$.

It follows easily from the definition of the protocol and \cref{lem:binary-random-split} that:

\begin{lemma}~
\label{lem:conditional-ubs}
\begin{enumerate}[label=(\roman*),leftmargin=1cm]
    \item For $1 \leq j \leq m-k/2$, $E[X_j| T=\vtot] \leq 3/2$
    \item $E[X_m|T=\vtot] \leq \lfloor t/2m \rfloor$
    \item $E[X_0|T=\vtot] \leq k$
\end{enumerate}
\end{lemma}

The following lemma requires analysing a great many cases, for which there is no space here.

\begin{lemma}~
\label{lem:marginal-lbs}
\begin{enumerate}[label=(\roman*),leftmargin=1cm]
    \item For $1 \leq j \leq m-k/2$, $E[X_j] \geq k/4h$
    \item $m-k/2 < j < m+1$, $E[X_j] \geq \frac{\max\{m+1-j,\log_2 k\}}{2h}$
    \item $E[X_{m+1}] \geq 3(k-2 \log_2 k)/4h$
    \item $E[X_0] \geq k/8$
\end{enumerate}
\end{lemma}

\subsection{Inference step}

We now assume that a vault has received an amount $v$, an instance of a random variable $V$, and consider what information we gain about $T$. It is evident that $t \geq v$, so many values of $t$ can be ruled out. Nonetheless, here we show that the probability of any particular value of $t$ greater than $v$ is not substantially larger than the probability given by the prior.

\begin{theorem} For any $j,t$ with $Pr[T=t|V=2^{j+1}] > 0$ and either $1 \leq j < m+1$ or $t < 2^{m+1}$, we have:
\begin{align*}
& \Pr[T=\vtot|V=2^{j+1}]\leq \\
& \Pr[T=\vtot]\frac{3h}{\min \{k/2,\max\{m+1-j, \log_2 k\}\}}
\end{align*}
For $V=0$, we have:
$$\Pr[T=\vtot|V=0]\leq 8 \Pr[T=t]$$
For $V=2^m$, $t \geq 2^{m+1}$:
$$\Pr[T=\vtot|V=2^m]\leq\Pr[T=\vtot]\frac{4h\lfloor t/2^m \rfloor}{3(k-2 \log_2 k)}$$
\end{theorem}

\begin{proof}
By Bayes rule, for $v \in 0,1,\dots 2^{m}$:
$$\Pr[T=t|V=v] =\frac{\Pr[V=v|T=\vtot]\Pr[T=\vtot]}{\Pr[V=v]}$$

We are interested in upper bounding the ratio $\frac{\Pr[V=v|T=\vtot]}{\Pr[V=v]}$. This ratio is $\frac{E[X_j|T=\vtot]}{E[X_j]}$ for respective $j$, with $\Pr[V=v]=E[X_j]/k$. We may now simply substitute the upper bounds on $E[X_j|T=\vtot]$ given by \cref{lem:conditional-ubs}, and the lower bounds on $E[X_j]$ given by \cref{lem:marginal-lbs}, to obtain the upper bound on the ratio on each case.
\end{proof}

\subsection{Analysis}

Substantial anonymity is achieved if \vtot is not close to its maximum value. Conditioning on $\vtot > 2^{j+1}$ rules out $n \leq j$, and thus gives $\Pr[T=\vtot|T>2^{j+1}]=\Pr[T=\vtot]\frac{h}{h-j}$ 
Because of the choice of splitting procedure, receiving a piece of size $2^j+1$ gives $k/2$ possibilities for $n$, where $j+1 \leq n \leq j + k/2$. This is why a factor of $\frac{h}{k/2}$ appears in the inequality for this range. However for $j+k/2 \geq n > m+1$, pieces of size $2^m$ would be required, leaving fewer bits to use for splitting and fewer possible $n$ which satisfy $j > m-k/2$. Hence we note that for larger values of \vtot, the anonymity is worse. However, crucially, the number of possible values for $n$ never falls below $\log_2 k$.

\section{Limitations and Future Work}
\label{sec:future}

We have argued that \zclaim provides privacy in cross-chain transfers, but many of the benefits of interoperability remain unattainable if privacy is to be maintained.
Once on the issuing chain, it is likely that interacting with the chain in any way other than through plain transactions would require the user to first convert the wrapped shielded assets to transparent assets.
Private exchanges have been proposed in the literature~\cite{da2021kicking}, but to the authors' knowledge none currently exists in the blockchain landscape.

There exist also a number of limitations inherent to the current design of the protocol, such as the number of transactions required for one cross-chain transfer following the splitting strategy.
As pointed out in \cref{sec:splitting_strategy}, this in not a problem on the Zcash side.
However, it imposes a limitation on which blockchains can realistically function as an issuing chain, as transaction fees on some blockchains can be orders of magnitudes higher than Zcash's.

Furthermore, the protocol presents a bootstrapping problem, in that it depends on a sufficiently large number of vaults with enough liquidity in order to function properly, which in turn depend on serving enough requests to be profitable.

The way issue and redeem availability is currently awarded, i.e. via the vault proving that it has at least the equivalent of \vmax in free collateral or by it revealing that it holds no ZEC obligations, respectively, is non-optimal and may reveal information about the transacted amount.
Other approaches should be explored, such as issue and redeem requests being assigned to vaults based on on-chain randomness, whereupon the vault chooses whether to accept it or not.
Additionally, although the number of concurrent Issue or Redeem requests being served by a vault is currently limited to one, we think concurrent requests should be feasible with minor modifications.
It may also prove beneficial to allow vaults to set transaction fees themselves, allowing them to fend off network congestion and at the same time incentivising competitiveness.

Furthermore, although the protocol as presented in this paper assumes that the currency being wrapped is Zcash, it can be adapted to any implementation of Sapling on another chain.
Such implementations currently exist e.g. on Tezos~\cite{tezos2014whitepaper,saplingTezos} and there are a number of adaptations of the Zerocash protocol deployed in smart contracts, for instance on Ethereum~\cite{rondelet2019zeth} and Quorum~\cite{ZSLConsenSysquorumWikiGitHub-2020-11-27}.
The changes required for compatibility would mostly limit themselves to the relay system and note commitment verification, whereas the protocol logic is independent of any specific implementation.

\zclaim is currently being adapted to the more recent Orchard version of Zcash and is to be extended for usage with multi-asset shielded pool extensions of the Zcash protocol.

\section{Conclusion}
\label{sec:conclusion}

We have shown that it is possible to maintain the privacy-preserving qualities of the Sapling specification of Zcash in cross-chain transfers.
More generally speaking, we provide a scheme for a decentralised cross-chain transfer protocol that integrates with a privacy-oriented cryptocurrency.
We show that no single intermediary can infer the total amount transferred through the bridge. 

\section*{Acknowledgements}

The authors would like to thank Petar Tsankov for supervising the master's thesis which made this work possible, Jeff Burdges for his assistance with cryptography, Alfonso Cevallos for useful discussions related to the splitting strategy, and Elizabeth Herbert for editorial support.

\bibliographystyle{IEEEtran}
\bibliography{IEEEabrv,refs.bib}

\begin{thebibliography}{10}
\providecommand{\url}[1]{#1}
\csname url@samestyle\endcsname
\providecommand{\newblock}{\relax}
\providecommand{\bibinfo}[2]{#2}
\providecommand{\BIBentrySTDinterwordspacing}{\spaceskip=0pt\relax}
\providecommand{\BIBentryALTinterwordstretchfactor}{4}
\providecommand{\BIBentryALTinterwordspacing}{\spaceskip=\fontdimen2\font plus
\BIBentryALTinterwordstretchfactor\fontdimen3\font minus
  \fontdimen4\font\relax}
\providecommand{\BIBforeignlanguage}[2]{{%
\expandafter\ifx\csname l@#1\endcsname\relax
\typeout{** WARNING: IEEEtran.bst: No hyphenation pattern has been}%
\typeout{** loaded for the language `#1'. Using the pattern for}%
\typeout{** the default language instead.}%
\else
\language=\csname l@#1\endcsname
\fi
#2}}
\providecommand{\BIBdecl}{\relax}
\BIBdecl

\bibitem{schulte2019towards}
S.~Schulte, M.~Sigwart, P.~Frauenthaler, and M.~Borkowski, ``Towards blockchain
  interoperability,'' in \emph{International conference on business process
  management}.\hskip 1em plus 0.5em minus 0.4em\relax Springer, 2019, pp.
  3--10.

\bibitem{belchior2021survey}
R.~Belchior, A.~Vasconcelos, S.~Guerreiro, and M.~Correia, ``A survey on
  blockchain interoperability: Past, present, and future trends,'' \emph{ACM
  Computing Surveys (CSUR)}, vol.~54, no.~8, pp. 1--41, 2021.

\bibitem{zhang2019security}
R.~Zhang, R.~Xue, and L.~Liu, ``Security and privacy on blockchain,'' \emph{ACM
  Computing Surveys (CSUR)}, vol.~52, no.~3, pp. 1--34, 2019.

\bibitem{sanchez2020confidential}
A.~Sanchez, ``Confidential cross-blockchain exchanges: Designing a
  privacy-preserving interoperability scheme,'' Unpublished master's thesis,
  ETH Z\"urich, Dec. 2020.

\bibitem{zamyatin2019xclaim}
A.~Zamyatin, D.~Harz, J.~Lind, P.~Panayiotou, A.~Gervais, and W.~Knottenbelt,
  ``{XCLAIM}: Trustless, interoperable, cryp\-to\-cur\-ren\-cy-backed assets,''
  in \emph{2019 IEEE Symposium on Security and Privacy (SP)}.\hskip 1em plus
  0.5em minus 0.4em\relax IEEE, Mar. 2019, pp. 193--210.

\bibitem{hopwood2016zcash}
\BIBentryALTinterwordspacing
D.~Hopwood, S.~Bowe, T.~Hornby, and N.~Wilcox, ``Zcash protocol
  specification,'' Electric Coin Company, Technical specification, Version
  2020.1.15 [Overwinter+Sapling], 2020. [Online]. Available:
  \url{https://zips.z.cash/protocol/sapling.pdf}
\BIBentrySTDinterwordspacing

\bibitem{banerjee2020demystifying}
A.~Banerjee, M.~Clear, and H.~Tewari, ``Demystifying the role of {zk-SNARKs} in
  {Zcash},'' in \emph{2020 IEEE Conference on Application, Information and
  Network Security (AINS)}.\hskip 1em plus 0.5em minus 0.4em\relax IEEE, 2020,
  pp. 12--19.

\bibitem{feng2019survey}
Q.~Feng, D.~He, S.~Zeadally, M.~K. Khan, and N.~Kumar, ``A survey on privacy
  protection in blockchain system,'' \emph{Journal of Network and Computer
  Applications}, vol. 126, pp. 45--58, 2019.

\bibitem{herlihy2018accs}
M.~Herlihy, ``Atomic cross-chain swaps,'' in \emph{Proceedings of the 2018 ACM
  Symposium on Principles of Distributed Computing}, ser. PODC '18.\hskip 1em
  plus 0.5em minus 0.4em\relax New York, NY, USA: Association for Computing
  Machinery, 2018, pp. 245--254.

\bibitem{Lightnin81:online}
\BIBentryALTinterwordspacing
[lightning-dev] an argument for single-asset lightning network. [Online].
  Available:
  \url{https://lists.linuxfoundation.org/pipermail/lightning-dev/2018-December/001752.html}
\BIBentrySTDinterwordspacing

\bibitem{van2013cryptonote}
\BIBentryALTinterwordspacing
N.~Van~Saberhagen, ``Cryptonote v 2.0,'' Monero white paper, 2013. [Online].
  Available:
  \url{https://www.getmonero.org/ru/resources/research-lab/pubs/whitepaper_annotated.pdf}
\BIBentrySTDinterwordspacing

\bibitem{BTCXMRatomicswaps}
\BIBentryALTinterwordspacing
J.~Gugger, ``Bitcoin-monero cross-chain atomic swap,'' Cryptology ePrint
  Archive, Report 2020/1126, 2020. [Online]. Available:
  \url{https://ia.cr/2020/1126}
\BIBentrySTDinterwordspacing

\bibitem{CCSMoneroAtomicSwapsimplementationfunding}
\BIBentryALTinterwordspacing
{CCS} - {Monero} atomic swaps. [Online]. Available:
  \url{https://ccs.getmonero.org/proposals/h4sh3d-atomic-swap-implementation.html}
\BIBentrySTDinterwordspacing

\bibitem{smoothie}
\BIBentryALTinterwordspacing
noot/atomic-swap: {ETH-XMR} atomic swap prototype. [Online]. Available:
  \url{https://github.com/noot/atomic-swap}
\BIBentrySTDinterwordspacing

\bibitem{zamyatin2019sok}
A.~Zamyatin, M.~Al-Bassam, D.~Zindros, E.~Kokoris-Kogias, P.~Moreno-Sanchez,
  A.~Kiayias, and W.~J. Knottenbelt, ``{SoK}: Communication across distributed
  ledgers,'' in \emph{Financial Cryptography and Data Security}.\hskip 1em plus
  0.5em minus 0.4em\relax Springer Berlin Heidelberg, 2021, pp. 3--36.

\bibitem{Wrapped}
\BIBentryALTinterwordspacing
Wrapped. [Online]. Available: \url{https://www.wrapped.com/}
\BIBentrySTDinterwordspacing

\bibitem{HomerenprojectrenWikiGitHub}
\BIBentryALTinterwordspacing
Home · renproject/ren wiki. [Online]. Available:
  \url{https://github.com/renproject/ren/wiki}
\BIBentrySTDinterwordspacing

\bibitem{Zcash76:online}
\BIBentryALTinterwordspacing
Zcash | {Ren} client docs. [Online]. Available:
  \url{https://renproject.github.io/ren-js-v3-docs/classes/_renproject_chains_bitcoin.Zcash.html}
\BIBentrySTDinterwordspacing

\bibitem{pegzone_announ}
\BIBentryALTinterwordspacing
Z.~F. Team. Bringing privacy to {Cosmos} with {Zcash}. [Online]. Available:
  \url{https://zfnd.org/bringing-privacy-to-cosmos-with-zcash/}
\BIBentrySTDinterwordspacing

\bibitem{githubPegzone}
\BIBentryALTinterwordspacing
Github - {ZcashFoundation}/zcash-pegzone: A shielded pegzone bridging {Cosmos}
  and {Zcash}. [Online]. Available:
  \url{https://github.com/ZcashFoundation/zcash-pegzone}
\BIBentrySTDinterwordspacing

\bibitem{cosmosWhitepaper}
\BIBentryALTinterwordspacing
J.~Kwon and E.~Buchman, ``Cosmos: A network of distributed ledgers,''
  Tendermint Inc., White paper. [Online]. Available:
  \url{https://v1.cosmos.network/resources/whitepaper}
\BIBentrySTDinterwordspacing

\bibitem{sasson2014zerocash}
E.~Ben-Sasson, A.~Chiesa, C.~Garman, M.~Green, I.~Miers, E.~Tromer, and
  M.~Virza, ``Zerocash: Decentralized anonymous payments from {Bitcoin},'' in
  \emph{2014 IEEE Symposium on Security and Privacy}.\hskip 1em plus 0.5em
  minus 0.4em\relax IEEE, May 2014, pp. 459--474.

\bibitem{sasson2014zerocash_ext}
\BIBentryALTinterwordspacing
------, ``Zerocash: Decentralized anonymous payments from {Bitcoin} (extended
  version),'' Cryptology ePrint Archive: Report 2014/349, May 2014. [Online].
  Available: \url{https://eprint.iacr.org/2014/349}
\BIBentrySTDinterwordspacing

\bibitem{nakamoto2008bitcoin}
\BIBentryALTinterwordspacing
S.~Nakamoto, ``Bitcoin: A peer-to-peer electronic cash system,'' Bitcoin white
  paper, 2008. [Online]. Available: \url{https://bitcoin.org/bitcoin.pdf}
\BIBentrySTDinterwordspacing

\bibitem{buterin2016interop}
\BIBentryALTinterwordspacing
V.~Buterin, ``Chain interoperability,'' R3, Tech. Rep., 2016. [Online].
  Available:
  \url{https://www.r3.com/wp-content/uploads/2017/06/chain_interoperability_r3.pdf}
\BIBentrySTDinterwordspacing

\bibitem{SPVBitcoinWiki}
\BIBentryALTinterwordspacing
Simplified payment verification ({SPV}) – {BitcoinWiki}. [Online]. Available:
  \url{https://en.bitcoinwiki.org/wiki/SPV}
\BIBentrySTDinterwordspacing

\bibitem{Back2014sidechains}
\BIBentryALTinterwordspacing
A.~Back, M.~Corallo, L.~Dashjr, M.~Friedenbach, G.~Maxwell, A.~Miller,
  A.~Poelstra, J.~Tim{\'o}n, and P.~Wuille, ``Enabling blockchain innovations
  with pegged side\-chains,'' Tech. Rep., 2014. [Online]. Available:
  \url{https://blockchainlab.com/pdf/sidechains.pdf}
\BIBentrySTDinterwordspacing

\bibitem{burdges2020overview}
J.~Burdges, A.~Cevallos, P.~Czaban, R.~Habermeier, S.~Hosseini, F.~Lama, H.~K.
  Alper, X.~Luo, F.~Shirazi, A.~Stewart \emph{et~al.}, ``Overview of {Polkadot}
  and its design considerations,'' Web3 Foundation, Tech. Rep., May 2020.

\bibitem{dfinity2022internet}
\BIBentryALTinterwordspacing
{The DFINITY Team}, ``The internet computer for geeks,'' Cryptology ePrint
  Archive, Report 2022/087, 2022. [Online]. Available:
  \url{https://ia.cr/2022/087}
\BIBentrySTDinterwordspacing

\bibitem{DepositProcessingTimesKraken}
\BIBentryALTinterwordspacing
Cryptocurrency deposit processing times – {Kraken}. [Online]. Available:
  \url{https://support.kraken.com/hc/en-us/articles/203325283-Cryptocurrency-deposit-processing-times}
\BIBentrySTDinterwordspacing

\bibitem{DepositProcessingTimesGemini}
\BIBentryALTinterwordspacing
How long until my crypto deposit reaches my account? – {Gemini}. [Online].
  Available:
  \url{https://support.gemini.com/hc/en-us/articles/205424836-How-long-until-my-digital-asset-deposit-reaches-my-account}
\BIBentrySTDinterwordspacing

\bibitem{quesnelle2017linkability}
J.~Quesnelle, ``On the linkability of {Zcash} transactions,'' Dec. 2017,
  arXiv:1712.01210 [cs.CR].

\bibitem{gervais2016security}
A.~Gervais, G.~O. Karame, K.~W{\"u}st, V.~Glykantzis, H.~Ritzdorf, and
  S.~Capkun, ``On the security and performance of proof of work blockchains,''
  in \emph{Proceedings of the 2016 ACM SIGSAC conference on computer and
  communications security}, 2016, pp. 3--16.

\bibitem{heilman2015eclipse}
E.~Heilman, A.~Kendler, A.~Zohar, and S.~Goldberg, ``Eclipse attacks on
  {Bitcoin}{\textquoteright}s peer-to-peer network,'' in \emph{24th {USENIX}
  Security Symposium ({USENIX} Security 15)}.\hskip 1em plus 0.5em minus
  0.4em\relax Washington, D.C.: {USENIX} Association, Aug. 2015, pp. 129--144.

\bibitem{wust2016ethereum}
K.~W{\"u}st and A.~Gervais, ``Ethereum eclipse attacks,'' ETH Zurich, Tech.
  Rep., 2016.

\bibitem{xu2020eclipsed}
G.~Xu, B.~Guo, C.~Su, X.~Zheng, K.~Liang, D.~Wong, and H.~Wang, ``Am {I}
  eclipsed? {A} smart detector of eclipse attacks for {Ethereum},''
  \emph{Computers \& Security}, vol.~88, p. 101604, Sep. 2019.

\bibitem{alangot2020decentralized}
B.~Alangot, D.~Reijsbergen, S.~Venugopalan, and P.~Szalachowski,
  ``Decentralized lightweight detection of eclipse attacks on {Bitcoin}
  clients,'' in \emph{2020 IEEE International Conference on Blockchain
  (Blockchain)}, 2020, pp. 337--342.

\bibitem{peterson2015augur}
\BIBentryALTinterwordspacing
J.~Peterson and J.~Krug, ``Augur: a decentralized, open-source platform for
  prediction markets,'' \emph{CoRR}, vol. abs/1501.01042, 2015. [Online].
  Available: \url{http://arxiv.org/abs/1501.01042}
\BIBentrySTDinterwordspacing

\bibitem{ellis2017chainlink}
\BIBentryALTinterwordspacing
S.~Ellis, A.~Juels, and S.~Nazarov, ``Chainlink: A decentralized oracle
  network,'' SmartContract ChainLink Ltd., White paper, 2017. [Online].
  Available: \url{https://research.chain.link/whitepaper-v1.pdf}
\BIBentrySTDinterwordspacing

\bibitem{lo2020reliabilityoracles}
S.~K. Lo, X.~Xu, M.~Staples, and L.~Yao, ``Reliability analysis for blockchain
  oracles,'' \emph{Computers \& Electrical Engineering}, vol.~83, p. 106582,
  Feb. 2020.

\bibitem{fry2016negative}
J.~Fry and J.~E.-T. Cheah, ``Negative bubbles and shocks in cryptocurrency
  markets,'' \emph{International Review of Financial Analysis}, vol.~47, pp.
  343--352, Feb. 2016.

\bibitem{cheah2015speculative}
J.~E.-T. Cheah and J.~Fry, ``Speculative bubbles in {Bitcoin} markets? an
  empirical investigation into the fundamental value of {Bitcoin},''
  \emph{Economics Letters}, vol. 130, pp. 32--36, Feb. 2015.

\bibitem{ZcashFea8:online}
\BIBentryALTinterwordspacing
Zcash feature {UX} checklist — {Zcash} documentation. [Online]. Available:
  \url{https://zcash.readthedocs.io/en/latest/rtd_pages/ux_wallet_checklist.html#transactions}
\BIBentrySTDinterwordspacing

\bibitem{da2021kicking}
\BIBentryALTinterwordspacing
M.~B. da~Gama, J.~Cartlidge, A.~Polychroniadou, N.~P. Smart, and Y.~T. Alaoui,
  ``Kicking-the-bucket: Fast privacy-preserving trading using buckets,''
  Cryptology ePrint Archive, Report 2021/1549, 2021. [Online]. Available:
  \url{https://ia.cr/2021/1549}
\BIBentrySTDinterwordspacing

\bibitem{tezos2014whitepaper}
\BIBentryALTinterwordspacing
L.~Goodman, ``Tezos: a self-amending crypto-ledger,'' White paper, 2014.
  [Online]. Available: \url{https://tezos.com/whitepaper.pdf}
\BIBentrySTDinterwordspacing

\bibitem{saplingTezos}
\BIBentryALTinterwordspacing
Nomadic {Labs} - sapling integration in {Tezos} - tech preview. [Online].
  Available:
  \url{https://research-development.nomadic-labs.com/sapling-integration-in-tezos-tech-preview.html}
\BIBentrySTDinterwordspacing

\bibitem{rondelet2019zeth}
A.~Rondelet and M.~Zajac, ``{ZETH}: On integrating {Zerocash} on {Ethereum},''
  Apr. 2019, arXiv:1904.00905 [cs.CR].

\bibitem{ZSLConsenSysquorumWikiGitHub-2020-11-27}
\BIBentryALTinterwordspacing
{ZSL} · {ConsenSys}/quorum wiki · {GitHub}. [Online]. Available:
  \url{https://github.com/ConsenSys/quorum/wiki/ZSL}
\BIBentrySTDinterwordspacing

\end{thebibliography}

\end{document}